\documentclass[double]{article}
\usepackage{times}
\usepackage{algorithm}
\usepackage{algorithmic}
\usepackage{wrapfig}
\usepackage{verbatim}
\usepackage{amsmath}
\usepackage{graphicx}
\usepackage{subcaption}
\usepackage[active]{srcltx}
\usepackage{color}
\usepackage{amsthm}
\usepackage{lineno}
\usepackage{dirtytalk}
\usepackage{amsthm}
\usepackage{authblk}
\usepackage{listings}
\usepackage{tikz}
\usepackage{longtable}
\usetikzlibrary{shapes.geometric, arrows}

\newtheorem{proposition}{Proposition}

\usepackage{epsfig,graphicx,amsfonts}
\usepackage{amsmath,amsthm,amssymb}
\usepackage{xcolor}

\usepackage{adjustbox}
\usepackage{multirow}
\usepackage{setspace}
\usepackage{hyperref}
\usepackage{comment}

\long\def\comment #1\commentend{}

\begin{document}

\title{\Large The Simplicity Paradox: Why Evolution Does Not Produce Universally Complex Agents}

\author{Teddy Lazebnik$^{1,2,*}$\\ \(^1\) Department of Information Systems, University of Haifa, Haifa, Israel \\ \(^2\) Department of Computing, Jonkoping University, Jonkoping, Sweden \\ \(^*\) Corresponding author: \url{teddy.lazebnik@ju.se} }

\date{ }

\maketitle 

\begin{abstract}
\noindent
It has been well established that information improves decisions, pushing the population forward as more information becomes available. Nevertheless, a wide range of empirical evidence shows that humans avoid complexity, delegate judgement, and prefer simplified social worlds. This tension raises an evolutionary puzzle: if knowledge is economically valuable and therefore evolutionarily beneficial, why do populations not converge towards universally informed and complex agents? In this study, we propose a theory of cognitive economy in which information has positive utility but costly acquisition, processing, and coordination. In complex environments, selection can favour heterogeneous populations: most individuals use low-cost heuristics and simplified choice architectures, whereas a minority of agents or institutions specialize in information processing. This cognitive division of labour reduces decision costs while preserving much of the value created by knowledge. We formalize this trade-off by comparing societies of uniformly complex agents with societies containing simpler agents and a specialized decision-making centre. The latter can dominate when the costs saved by distributed simplicity exceed the utility lost through reduced individual autonomy and imperfect delegation. Crucially, the specialized decision-maker need not face a volunteer’s dilemma, because its private payoff can exceed that available under universal complexity through rents, status, control or superior information. The framework links bounded rationality, rational inattention, hierarchy, markets, and cultural evolution, and suggests that simplicity is not a failure of adaptation but a precondition for scalable social organization.
 \\ \\

\noindent
\textbf{Keywords}: Cognitive economy; bounded rationality; information costs; cognitive labour; evolutionary economics.
\end{abstract}

\maketitle \thispagestyle{empty}
\pagestyle{myheadings} \markboth{Draft:  \today}{Draft:  \today}
\setcounter{page}{1}

\section{Introduction}
Information is often treated as a fundamental economic resource because more informative signals can expand the set of attainable decisions and weakly improve expected payoffs under standard decision-theoretic assumptions \cite{Blackwell1953,stiglitz1985information,schiller2024think}. In strategic environments, incomplete and asymmetric information change both the beliefs agents hold and the equilibria they can sustain \cite{Harsanyi1967}. In financial economics, the efficient-market tradition formalizes the idea that prices aggregate available information about assets \cite{Fama1970}. In evolutionary ecology, information about food, predators, mates, and social partners is similarly analysed as a resource that can affect survival and reproductive success \cite{Dall2005}.

The economic value of information, however, is conditional on the cost of acquiring and using it \cite{bloedel2025cost.weller2018does,shami2022economic}. Search theory shows that agents should not gather information indefinitely, but only while the expected marginal benefit of additional information exceeds its marginal cost \cite{Stigler1961}. Rational-inattention models extend this logic by treating attention and information-processing capacity as scarce constraints on choice \cite{Sims2003}. The impossibility of perfectly informationally efficient markets follows from the same principle: if information is costly, informed agents must be compensated, so prices cannot fully reveal all information while leaving no incentive to become informed \cite{GrossmanStiglitz1980}. In large collective decisions, costly information acquisition can also decline at the individual level as the population size increases, producing rational ignorance rather than universal knowledge acquisition \cite{Martinelli2006}.

Human cognition displays the same cost-sensitive structure \cite{gigerenzer2020bounded,geng2022human}. Bounded rationality replaces the fully optimizing agent with an agent whose decisions are constrained by limited information, limited computation, and limited time \cite{Simon1955}. Studies on short-term memory identified sharp limits on immediate information processing capacity \cite{Miller1956,oberauer2018benchmarks}. Behavioural decision research shows that individuals use heuristics under uncertainty rather than computing all probabilities and outcomes explicitly \cite{TverskyKahneman1974}. Heuristics can be understood as effort-reduction devices that reduce search, simplify representation, and decrease computation \cite{ShahOppenheimer2008}. Experimental evidence further indicates that people often avoid cognitively demanding actions even when they are not explicitly aware of the demand manipulation \cite{Kool2010}. Choice-overload experiments show that increasing the number of options can reduce motivation and action, suggesting that more information and more alternatives do not always improve effective choice \cite{IyengarLepper2000}.

These considerations suggest that social organization should not be understood only as the aggregation of individual decisions, but also as a mechanism for allocating cognitive costs. Organization design theory treats uncertainty as an information-processing burden and explains organizational structure as a response to that burden \cite{Galbraith1974}. The theory of the firm as a communication network shows that organizations can be designed to economize the costs of processing and transmitting information among agents \cite{BoltonDewatripont1994}. Knowledge-based hierarchy emerges naturally when common problems can be solved by many workers, but exceptional problems must be routed to specialized problem-solvers \cite{Garicano2000}. The distinction between formal and real authority further shows that decision rights and effective control depend on the distribution of information inside organizations \cite{AghionTirole1997}.

Institutions more broadly can be interpreted as devices that reduce uncertainty and structure repeated interaction \cite{North1990}. Markets perform one version of this function because prices communicate dispersed knowledge without requiring each participant to know the full state of the economy \cite{Hayek1945}. Hierarchies perform another version of this function because centralized influence can reduce scalar stress as group size increases \cite{PerretHartPowers2020}. Division of labour theory adds that specialization is limited not only by market size but also by coordination costs and the distribution of knowledge \cite{BeckerMurphy1992}.

This paper builds on these literatures to propose a theory of cognitive economy. The central puzzle is that information is valuable, yet populations do not converge towards universal cognitive complexity. We argue that this phenomenon is an adaptive response to the costs of information acquisition, processing, and coordination. In complex socioeconomic environments, populations may increase collective utility by allowing most agents to rely on low-cost heuristics while a minority of agents, roles, or institutions specialize in processing complexity. We call these structures decision-compression institutions because they translate complex social states into simpler signals, rules, prices, recommendations, commands or norms.

The proposed framework compares three social configurations: a population of uniformly complex agents, a population of uniformly simple agents, and a heterogeneous population in which simple agents coexist with a specialized decision-making centre. The model captures the trade-off between information benefits, cognitive costs, delegation losses, and institutional rents. It predicts that heterogeneous cognitive organization can dominate universal complexity when the costs saved by widespread simplicity exceed the utility lost through imperfect delegation. It further predicts that the specialized decision-maker need not face a volunteer’s dilemma, because the role can generate private benefits such as rents, status, authority, control or superior information \cite{Diekmann1985}. Therefore, the contribution of the paper is to reframe simplicity as a productive social force. Simplicity is often described as ignorance, passivity, or cognitive limitation. We argue instead that simplicity can be a precondition for scalable social organization. Complex societies do not require all individuals to become complex. They require mechanisms that decide where complexity should be processed, who should bear its costs, and how the outputs of complex processing should be translated into forms that ordinary agents can use.

The remainder of this paper is organized as follows. Section~\ref{sec:rw} reviews the related literature on costly information, bounded rationality, cognitive simplicity, and the emergence of decision-making institutions. Section~\ref{sec:model} introduces the proposed model and formally defines the population structures, utility components, and delegation mechanisms considered in the analysis. Section~\ref{sec:results} presents a numerical investigation of the model, examining the conditions under which heterogeneous cognitive organization and centralized decision-making outperform uniformly simple or uniformly complex populations. Finally, Section~\ref{sec:discussion} discusses the theoretical implications of the results, their relation to evolutionary economics and institutional theory, and possible directions for future research.

\section{Related Work}
\label{sec:rw}
The literature relevant to the proposed framework can be organized around three connected ideas. First, information has economic value because it improves decision quality, but this value is limited by the costs of acquiring, processing, and using information. Second, bounded rationality and rational inattention explain why agents often rely on simplified representations, heuristics, and selective attention rather than full optimization. Third, evolutionary and institutional perspectives show that cognition itself is costly, so populations and organizations may allocate information-processing tasks unevenly across individuals, roles, and institutions.

\subsection{The economic value and cost of information}
\label{subsec:info_value_cost}
Information has economic value because it changes what agents can know, which actions they can justify, and how well decisions can be matched to uncertain states of the world, but this value is always constrained by the costs of acquiring, processing, transmitting, and using information \cite{stiglitz1985information}. In this sense, information is not simply a stock of facts, but an input whose usefulness depends on the social and technical systems that make it interpretable and actionable \cite{schiller2024think}. Thus, more information can improve decisions, yet beyond some point, the marginal value of additional information may be smaller than the marginal cognitive, institutional, or strategic cost of obtaining it \cite{Stigler1961}. This trade-off is especially important in complex environments, where agents face many possible states, limited attention, and costly consequences of mistaken decisions \cite{Sims2003}.

In its core, information improves decision quality by refining the signal on which action is based \cite{Blackwell1953}. Blackwell's comparison of experiments formalized the intuition that a more informative signal structure can only improve expected outcomes for a rational decision-maker who can choose optimally after observing the signal \cite{Blackwell1953}. In strategic settings, information also affects beliefs about other agents, so incomplete information changes the structure of equilibrium and the range of rational behavior \cite{Harsanyi1967}. In financial markets, the value of information appears through prices, because prices partly summarize dispersed beliefs about future returns and risks \cite{Fama1970}. In evolutionary ecology, information about predators, resources, mates, and competitors is similarly valuable because it can alter survival and reproductive decisions \cite{Dall2005}.

On the other hand, information is costly, so agents should not acquire all available knowledge \cite{GrossmanStiglitz1980}. Search theory shows that rational agents stop searching when the expected improvement from another observation is lower than the cost of obtaining it \cite{Stigler1961}. Rational-inattention models extend this logic by treating attention itself as scarce, meaning that agents optimally choose simplified representations of the world rather than processing all signals with equal precision \cite{Sims2003}. The Grossman--Stiglitz paradox shows that perfectly informative prices cannot exist in equilibrium when information acquisition is costly, because informed agents must receive some compensation for becoming informed \cite{GrossmanStiglitz1980}. In collective choice, costly information can produce rational ignorance, because a single individual's probability of changing the outcome is often too small to justify the private cost of becoming fully informed \cite{Martinelli2006}.

Importantly, the cost of information is not only monetary but also cognitive and biological \cite{Simon1955}. Bounded rationality argues that real agents satisfice because they face limits on time, memory, and computation rather than unlimited optimizing capacity \cite{Simon1955}. Empirical work on short-term memory shows that human information processing has sharp capacity constraints, making simplification a structural feature of cognition rather than merely a behavioral mistake \cite{Miller1956}. The heuristics-and-biases tradition shows that people often use simplified judgment rules under uncertainty, sometimes efficiently and sometimes with systematic error \cite{TverskyKahneman1974}. Later work reframed heuristics as effort-reduction mechanisms that reduce search, simplify choice, and make decisions possible under limited cognitive resources \cite{ShahOppenheimer2008}. Experiments on cognitive demand show that people often avoid mentally demanding choices even when effort costs are indirect or implicit \cite{Kool2010}. Choice-overload studies further show that more options and more information can reduce motivation and action, demonstrating that informational abundance can become behaviorally costly \cite{IyengarLepper2000}. 

Further in this line, information processing is biologically expensive, so evolution should favor economical cognition rather than universal complexity \cite{AielloWheeler1995}. The expensive-tissue hypothesis links large brains to metabolic trade-offs, suggesting that cognitive capacity must be paid for by energy reallocation elsewhere in the organism \cite{AielloWheeler1995}. Neuroscientific evidence likewise shows that neural signaling has direct metabolic costs, so representing and transmitting information is physically expensive \cite{Laughlin1998}. Evolutionary models of memory emphasize that storing information can create maintenance costs, interference costs, and opportunity costs \cite{Dukas1999}. Fast-and-frugal models therefore interpret simple rules not as inferior approximations in every case, but as adaptive tools suited to environments where information and computation are costly \cite{GigerenzerGoldstein1996}.

Based on these processes, societies respond to information costs through institutions, markets, firms, and hierarchies \cite{Galbraith1974}. Indeed, organization design theory treats uncertainty as an information-processing burden and explains organizational structure as a mechanism for matching information requirements with processing capacity \cite{Galbraith1974}. The firm can be modeled as a communication network in which tasks are allocated so as to economize on the costs of transmitting and processing knowledge \cite{BoltonDewatripont1994}. Knowledge-based hierarchy emerges when ordinary problems can be solved locally but exceptional problems are routed to specialized experts \cite{Garicano2000}. The distinction between formal and real authority shows that control depends not only on legal decision rights but also on who possesses the information needed to make effective decisions \cite{AghionTirole1997}. Markets solve part of the same problem because prices compress dispersed knowledge into signals that can guide agents without requiring them to know the full economy \cite{Hayek1945}. Institutions more generally reduce uncertainty by stabilizing expectations, rules, and incentives across repeated interactions \cite{North1990}. Division of labour expands the productive use of knowledge, but it is limited by coordination costs that grow as specialized agents become more interdependent \cite{BeckerMurphy1992}. Recent evolutionary models of hierarchy similarly show that larger groups may benefit from centralized structure when coordination pressure rises with group size \cite{PerretHartPowers2020}. Modern market microstructure adds that new technologies can also change the incentives to acquire information, as algorithmic trading may reduce some forms of information acquisition while increasing the value of speed and processing capacity \cite{weller2018does}.

\subsection{Bounded rationality, rational inattention, and cognitive effort}
\label{subsec:bounded_rationality}
Theories of bounded rationality, rational inattention, and cognitive effort all begin from the same central observation: agents do not face decision problems only through preferences and constraints, but also through limited capacities for attention, representation, memory, computation, and control \cite{March1978}. This shifts the analysis of choice away from the image of a fully optimizing actor and toward an actor who must decide how much of the world to process before acting \cite{Conlisk1996}. The resulting behavior is not necessarily irrational; rather, it reflects the fact that the process of becoming informed, comparing alternatives, and exerting mental control is itself costly \cite{McFadden1999}. Evidence from psychology and economics shows that these limits matter in real environments, including consumer markets, financial choices, labour decisions, and policy responses \cite{DellaVigna2009}.

Bounded rationality originally framed decision-making as a procedural activity carried out under complexity rather than as the direct solution of a complete optimization problem \cite{CyertMarch1963}. In this view, agents search, simplify, categorize, and stop once a sufficiently good option has been found, because exhaustive comparison is often infeasible \cite{NewellSimon1972}. Decision procedures therefore adapt to task demands: when a problem is simple, agents may use more compensatory and information-rich strategies, but when complexity increases, they tend to rely on selective attention, simplifying rules, and reduced comparison sets \cite{PayneBettmanJohnson1993}. This is why heuristics can be understood not merely as cognitive defects, but as tools that trade some potential accuracy for major savings in effort and computation \cite{GigerenzerSelten2001}. The same logic appears in strategic settings, where individuals frequently reason with limited depth rather than computing full equilibrium behavior \cite{CamererHoChong2004}.

Rational inattention formalizes this intuition by treating attention as a scarce resource that must be allocated optimally \cite{Gabaix2014}. Instead of assuming that all available information is freely observed and processed, rational-inattention models ask which signals an agent should acquire when each additional bit of precision has a cost \cite{MatejkaMcKay2015}. This approach explains why choice can appear noisy even when agents are purposeful: stochastic choice may emerge because agents optimally choose imperfect information about alternatives \cite{CaplinDean2015}. It also explains why individuals may not consider all available options, since the formation of a consideration set is itself an informational decision \cite{CaplinDeanLeahy2019}. More broadly, behavioral inattention connects many empirical deviations from standard rationality to the same underlying mechanism: agents economize on attention when the world presents more information than they can profitably process \cite{Gabaix2019}. Perceptual models make a similar point by showing that valuation errors may arise from efficient but capacity-limited encoding of economic stimuli \cite{Woodford2012}.

The empirical importance of these constraints is visible in cases where equivalent information produces different behavior depending on how it is presented. Consumers respond differently to taxes when they are salient at the point of purchase, showing that economic choices depend not only on objective prices but also on whether relevant information enters attention \cite{ChettyLooneyKroft2009}. In the used-car market, left-digit bias in odometer readings shows that even high-stakes transactions are affected by coarse mental representations and focal cues \cite{LaceteraPopeSydnor2012}. Individual differences in cognitive reflection further show that people vary in their willingness or ability to override intuitive responses, and these differences are associated with risk-taking and intertemporal choice \cite{Frederick2005}. These findings suggest that bounded rationality is not simply a theoretical relaxation of optimization, but an empirically measurable feature of economic behavior.

Cognitive effort provides the motivational foundation for this pattern. Attention and control are not only limited; they are also experienced as costly, which means agents may avoid demanding computations even when they could improve accuracy \cite{Kahneman1973}. Neuroeconomic work treats cognitive effort as a decision variable, where people compare the expected benefit of exerting control with the subjective cost of doing so \cite{WestbrookBraver2015}. The expected-value-of-control framework develops this idea by arguing that control is allocated when its anticipated benefits exceed the costs of implementation \cite{ShenhavBotvinickCohen2013}. Motivation, reward, and control are therefore tightly linked, because incentives can increase the willingness to engage costly cognition \cite{BotvinickBraver2015}. Opportunity-cost theories add that effort feels costly partly because using executive resources for one task prevents their use elsewhere \cite{KurzbanDuckworthKableMyers2013}.

These perspectives imply that simplicity is not an accidental failure of decision-making but a predictable response to costly cognition. Resource-rational theories make this explicit by evaluating cognitive mechanisms according to how well they use limited computational resources, rather than by comparing them to unconstrained optimization \cite{GriffithsLiederGoodman2015}. In this framework, a simple rule can be rational when the additional value of a more complex calculation is smaller than the cost of performing it \cite{LiederGriffiths2020}. Information-theoretic models reach a similar conclusion by treating bounded rationality as utility maximization under information-processing constraints \cite{OrtegaBraunDyerKimTishby2015}. The general implication is that cognitive economy is a fundamental part of adaptation: agents simplify because attention, computation, and control are scarce, and complex societies must therefore organize not only the distribution of material resources but also the distribution of cognitive effort \cite{BossaertsMurawski2017}.

\subsection{Evolutionary costs of cognition and information processing}
\label{subsec:evolutionary_cognition_costs}
The evolutionary cost of cognition can be understood in economic terms as the cost of producing, maintaining, and using an information-processing asset under scarcity \cite{Alchian1950}. From this perspective, cognition is not free intelligence added to an organism or agent, but a capital-intensive capacity that must generate returns large enough to justify its acquisition and operating costs \cite{NelsonWinter1982}. Human cognitive expansion is therefore often interpreted as an investment problem in which longer learning periods, skill accumulation, and delayed returns become viable only when the expected productivity of information processing is sufficiently high \cite{RobsonKaplan2003}. Life-history accounts make the same point by arguing that intelligence, longevity, and complex foraging coevolve when knowledge-intensive production can compensate for the costs of extended development \cite{KaplanHillLancasterHurtado2000}.

The most direct cost of cognition is the resource cost of building and operating neural infrastructure \cite{IslerVanSchaik2009}. Larger brains require energetic support, so selection should favor cognitive expansion only when additional information-processing capacity increases net fitness or productivity after metabolic costs are deducted \cite{NavarreteVanSchaikIsler2011}. Neural signaling itself is expensive, which means that cognition has a continuing operating cost rather than only a one-time construction cost \cite{AttwellLaughlin2001}. Communication between neurons must also economize on bandwidth, precision, and reliability, creating a biological analogue of costly information transmission in economic systems \cite{LaughlinSejnowski2003}. Sensory systems face the same constraint because collecting more detailed information from the environment can improve action, but only at the price of higher acquisition and processing costs \cite{NivenLaughlin2008}. Across vertebrates, the relative metabolic burden of the central nervous system suggests that cognitive capacity is constrained by a broad energy budget rather than by informational usefulness alone \cite{MinkBlumenschineAdams1981}.

These costs imply that larger or more complex cognition is not automatically superior, because returns to information processing are environment-dependent and may exhibit diminishing marginal value \cite{ChittkaNiven2009}. Experimental evolution provides direct support for this trade-off: selection for learning ability in fruit flies can generate fitness costs, indicating that improved information processing may reduce performance along other dimensions \cite{MeryKawecki2003}. Artificial selection experiments in guppies similarly show that larger brains can improve cognitive performance while imposing costs on reproduction or other fitness-related traits \cite{Kotrschal2013}. In economic language, cognition therefore involves opportunity costs: resources allocated to information processing cannot simultaneously be allocated to growth, reproduction, physical capacity, or other adaptive investments \cite{NivenFarris2008}. The relevant question is not whether complex cognition is useful, but whether its marginal benefit exceeds its marginal cost in the environment in which the agent operates \cite{GodfreySmith2002}.

The value of cognition rises when environments create high returns to forecasting, coordination, innovation, and social inference \cite{DunbarShultz2007}. Social-brain theories argue that complex groups increase the payoff from tracking relationships, intentions, coalitions, and reputations, making cognition valuable as a coordination technology \cite{ByrneWhiten1988}. Comparative studies show that innovation, social learning, and brain size are associated across primates, suggesting that cognition can function as an adaptive response to information-rich social and ecological problems \cite{ReaderLaland2002}. In birds, larger brains are associated with improved responses to novel environments, which suggests that information-processing capacity can act as a form of flexible capital under uncertainty \cite{SolDuncanBlackburnCasseyLefebvre2005}. Yet flexibility remains costly, so selection should favor it mainly when environmental variability makes fixed routines less profitable \cite{vanSchaikBurkart2011}.

A crucial economic implication is that cognition can be partly substituted by social learning, institutions, and culture \cite{BoydRicherson1985}. Cultural transmission allows agents to acquire useful behavioral rules without independently paying the full cost of discovery, making society a mechanism for amortizing information costs across individuals and generations \cite{HenrichMcElreath2003}. This can generate cumulative knowledge, but only when populations are large and connected enough to preserve and improve complex skills \cite{Henrich2004}. Models and experiments on social learning show that copying others can be efficient when private information is costly, uncertain, or slow to acquire \cite{Rendell2010}. Demographic studies of technological complexity likewise suggest that larger populations can sustain more complex toolkits because they distribute the costs of innovation, imitation, and error correction over more learners \cite{KlineBoyd2010}. In this sense, culture acts as an external memory system that reduces the need for every individual to rediscover the same information independently \cite{Tomasello1999}. However, social learning also creates dependence on transmitted signals, so populations must balance the savings from imitation against the risk of copying outdated or maladaptive information \cite{Rogers1988}.

When considered as a whole, the evolutionary literature supports an economic interpretation of cognitive simplicity. Complex cognition is valuable when it increases the expected return from action, but it is costly to build, costly to operate, and costly to coordinate \cite{SterlingLaughlin2015}. Selection should therefore not be expected to produce universally complex agents, because in many environments the efficient solution is a portfolio of cognitive strategies: some information is processed individually, some is delegated to specialists, and some is stored in social rules, tools, markets, or institutions \cite{Laland2017}. Cognitive economy emerges from this allocation problem: evolution favors not maximal information processing, but information processing whose benefits exceed its full cost \cite{RichersonBoyd2005}. This outcome is the foundation of the modeling attempt presented in this study.

\section{Model Definition}
\label{sec:model}
Consider a population of size \((N \geq 2\) of agents. The population operates in an environment with complexity level \(e \geq 0\). Higher values of \(e\) represent environments with greater uncertainty, more interdependence among agents, more possible actions, or more costly consequences of poor decisions. Here, each agent receives a decision-relevant payoff. We abstract from the specific action space and represent the expected gross benefit of decision-making by the function \(B(e,q)\), where \(q \in [0,1]\) denotes the quality of the information used in the decision. We assume that \(B(e,q)\) satisfies:
\begin{equation}
\frac{\partial B}{\partial q} > 0,
\qquad
\frac{\partial B}{\partial e} \geq 0,
\qquad
\frac{\partial^2 B}{\partial e \partial q} \geq 0.
\end{equation}
Namely, better information improves expected payoff, and the marginal value of information is weakly increasing in environmental complexity.

In order to keep the model simple while still expressive enough to capture the dynamics, let us assume that agents operate under one of two cognitive modes: simple or complex. A simple agent uses low-cost heuristics, routines, norms, or external cues. A complex agent acquires and processes more information independently. Let \(q_S\) and \(q_C\) denote the information quality available to simple and complex agents, respectively, with \(q_S<q_C\). Let \(K_S(e)\) and \(K_C(e)\) denote their corresponding cognitive costs, with \(K_S(e)<K_C(e)\). These costs include attention, learning, deliberation, computation, search, and coordination costs. We assume that the cost gap between complex and simple cognition is weakly increasing in environmental complexity:
\begin{equation}
\frac{\partial}{\partial e}\left(K_C(e)-K_S(e)\right) \geq 0.
\end{equation}
Simply put, this assumption captures the idea that complexity makes independent high-quality decision-making increasingly costly.

In a uniformly complex society, all agents independently acquire and process high-quality information. The per-agent utility is \(U_C(e)=B(e,q_C)-K_C(e)\). Total welfare in the uniformly complex society is therefore \(W_C(e,N)=N\left(B(e,q_C)-K_C(e)\right)\). This configuration maximizes individual autonomy and avoids delegation loss, but it requires every member of the population to bear the cost of complex cognition. In a uniformly simple society, all agents rely on low-cost cognition. The per-agent utility is \(U_S(e)=B(e,q_S)-K_S(e)\). Total welfare is \(W_S(e,N)=N\left(B(e,q_S)-K_S(e)\right)\). This configuration minimizes cognitive cost, but it may lose substantial utility when the environment is complex and the value of information is high. This is the simplest case. Next, we define a heterogeneous society with cognitive specialization. 

In the heterogeneous society with a cognitive specialization configuration, \(N-1\) agents remain simple, while one agent acts as a central decision-maker. The central decision-maker acquires and processes information of quality \(q_D\), where \(q_D \geq q_C\). The central decision-maker translates this information into a simplified output for the rest of the population. This output may take the form of a rule, recommendation, command, price, norm, organizational routine, or institutional signal. In this context, delegation is imperfect. Hence, let \(L(e,N) \geq 0\) denote the delegation loss incurred by each simple follower when using the output of the central decision-maker. This loss captures distortion, agency problems, imperfect communication, reduced autonomy, or a mismatch between the central decision and local conditions. We allow \(L(e,N)\) to increase with environmental complexity and population size. Each simple follower pays a transfer \(r \geq 0\) to the central decision-maker. This transfer may represent taxes, wages, fees, obedience rents, status benefits, or other forms of compensation. The follower utility in the decision-compression society is \(U_F(e,N,r)=B(e,q_D)-L(e,N)-K_S(e)-r\).

The central decision-maker bears the cognitive cost \(K_D(e)\), with \(K_D(e) \geq K_C(e)\), and may also bear an institutional overhead cost \(H(e,N) \geq 0\). This overhead captures the cost of communicating decisions, maintaining authority, monitoring followers, or operating the decision-compression institution. The central decision-maker receives transfers from the \(N-1\) followers and may obtain an additional private benefit \(\chi \geq 0\), representing status, authority, control, reputation, or superior access to information. The central decision-maker's utility is \(U_D(e,N,r,\chi)=B(e,q_D)-K_D(e)-H(e,N)+(N-1)r+\chi\). 

Ignoring pure transfers, total welfare in the decision-compression society is
\begin{equation}
W_D(e,N)=N B(e,q_D)-(N-1)L(e,N)-(N-1)K_S(e)-K_D(e)-H(e,N).
\label{eq:welfare}
\end{equation}

\subsection{Dominance of decision-compression}
The decision-compression society dominates the uniformly complex society when
\(W_D(e,N)>W_C(e,N)\). Substituting the welfare expressions gives
\begin{equation}
N\left(B(e,q_D)-B(e,q_C)\right)
+
N K_C(e)
-
(N-1)K_S(e)
-
K_D(e)
-
H(e,N)
>
(N-1)L(e,N).
\label{eq:dc_dominates_complex}
\end{equation}

This condition states that decision-compression is socially beneficial when
the informational advantage of the central decision-maker and the cognitive
costs saved by the simple followers exceed the delegation loss and
institutional overhead.

A particularly informative case is \(q_D=q_C\), where the central
decision-maker has the same information quality as a complex individual but
processes information once on behalf of others. In this case, the condition in Eq. 
\eqref{eq:dc_dominates_complex} becomes
\begin{equation}
N K_C(e)
-
(N-1)K_S(e)
-
K_D(e)
-
H(e,N)
>
(N-1)L(e,N).
\label{eq:dc_dominates_complex_equal_quality}
\end{equation}

Thus, even without superior information, central decision-making can be
welfare-improving when the savings from avoiding duplicated cognitive effort
exceed the costs of delegation and institutional operation.

The decision-compression society dominates the uniformly simple society when
\(W_D(e,N)>W_S(e,N)\). This is equivalent to
\begin{equation}
N\left(B(e,q_D)-B(e,q_S)\right)
>
(N-1)L(e,N)
+
K_D(e)
+
H(e,N)
-
K_S(e).
\label{eq:dc_dominates_simple}
\end{equation}

This condition states that decision-compression improves upon universal
simplicity when the informational benefit generated by the central
decision-maker exceeds delegation loss, central cognitive cost, and
institutional overhead.

\subsection{Individual rationality and the absence of a volunteer's dilemma}

The emergence of a central decision-maker may appear to create a volunteer's dilemma: one individual must bear a high cognitive and institutional cost while others benefit from simplified decisions. In this model, however, the central role need not be altruistic and therefore can emerge spontaneously. The central decision-maker voluntarily accepts the role when \(U_D(e,N,r,\chi) \geq U_C(e)\). Substituting the
utility functions gives
\begin{equation}
(N-1)r+\chi
\geq
B(e,q_C)-B(e,q_D)
+
K_D(e)-K_C(e)
+
H(e,N).
\label{eq:dm_participation}
\end{equation}

Therefore, the central role is individually rational when transfers, status,
authority, or control benefits compensate for the additional cost of
centralized decision-making.

Followers accept delegation when their utility under decision-compression is
at least as high as their utility under uniform complexity, namely
\(U_F(e,N,r) \geq U_C(e)\). This gives
\begin{equation}
B(e,q_D)
-
L(e,N)
-
K_S(e)
-
r
\geq
B(e,q_C)
-
K_C(e),
\end{equation}
or
\begin{equation}
r
\leq
B(e,q_D)-B(e,q_C)
+
K_C(e)-K_S(e)
-
L(e,N).
\label{eq:follower_participation}
\end{equation}

A feasible decision-compression institution exists when there is some transfer
\(r\) satisfying both the central decision-maker's participation condition and
the followers' participation condition. This requires
\begin{equation}
\frac{
B(e,q_C)-B(e,q_D)
+
K_D(e)-K_C(e)
+
H(e,N)
-
\chi
}{N-1}
\leq
B(e,q_D)-B(e,q_C)
+
K_C(e)-K_S(e)
-
L(e,N).
\label{eq:feasible_transfer}
\end{equation}

When this condition holds, the central decision-maker is not a volunteer who
sacrifices utility for the group. Rather, cognitive specialization is mutually
sustainable: followers gain by reducing cognitive costs, and the central
decision-maker gains through compensation, authority, or informational
advantage.

\subsection{Role-choice game and equilibrium centralization}

We now extend the model to a symmetric role-choice game in order to study
whether a central decision-maker can emerge endogenously. The agents are
ex ante symmetric: before roles are chosen, every agent has the same cognitive
ability, the same opportunity to become a central decision-maker, and the same
payoff functions. Symmetry therefore concerns the agents' primitives, not the
equilibrium outcome.

Each agent chooses one of two roles: follower or decision-maker. Let
\(k \in \{0,1,\ldots,N\}\) denote the number of agents who choose to become
decision-makers. If \(k=0\), no decision-compression institution exists and
each agent receives the complex baseline payoff
\begin{equation}
U_0(e)=B(e,q_C)-K_C(e).
\end{equation}

If \(k \geq 1\), followers use the decision-compression output and receive
\begin{equation}
U_F(e)=B(e,q_D)-L(e,N)-K_S(e)-r.
\end{equation}

In this configuration, decision-makers duplicate the costly information-processing
role and compete for followers. For analytical clarity, suppose that each
follower pays the transfer \(r\) to one decision-maker and that followers are
divided symmetrically among the \(k\) decision-makers. In addition, multiple
decision-makers create a duplication or conflict cost \(\eta(k-1)\), where
\(\eta \geq 0\), and dilute the private authority benefit by \(\psi(k-1)\),
where \(\psi \geq 0\). The payoff of each decision-maker when there are \(k\)
decision-makers is therefore
\begin{equation}
U_D(e,k)
=
B(e,q_D)
-
K_D(e)
-
H(e,N)
-
\eta(k-1)
+
\frac{N-k}{k}r
+
\chi
-
\psi(k-1).
\label{eq:dm_payoff_k}
\end{equation}

The term \((N-k)r/k\) is the transfer revenue received by each decision-maker.
It decreases as the number of decision-makers increases because the follower
population is divided among competing decision-makers.

An allocation with \(k\) decision-makers is stable if no single agent can
improve its utility by changing role. Thus, if \(k=0\), stability requires that
no agent wants to become a decision-maker. If \(1 \leq k \leq N-1\), stability
requires that no follower wants to become an additional decision-maker and no
decision-maker wants to become a follower. If \(k=N\), stability requires that
no decision-maker wants to become a follower.

\begin{proposition}
\label{prop:unique_central_dm}
Suppose that agents are ex ante symmetric and that the following conditions
hold:
\begin{equation}
B(e,q_D)
-
K_D(e)
-
H(e,N)
+
(N-1)r
+
\chi
>
B(e,q_C)
-
K_C(e),
\label{eq:entry_condition}
\end{equation}
\begin{equation}
B(e,q_D)
-
L(e,N)
-
K_S(e)
-
r
\geq
B(e,q_C)
-
K_C(e),
\label{eq:follower_participation_condition}
\end{equation}
and
\begin{equation}
K_D(e)
+
H(e,N)
+
\eta
+
\psi
-
L(e,N)
-
K_S(e)
-
\chi
>
\frac{N}{2}r.
\label{eq:no_duplication_condition}
\end{equation}
Then the unique stable equilibrium of the role-choice game contains exactly
one central decision-maker and \(N-1\) followers.
\end{proposition}

\begin{proof}
First consider the allocation with \(k=0\). In this allocation, every agent
receives \(U_0(e)=B(e,q_C)-K_C(e)\). If one agent deviates and becomes the
sole central decision-maker, that agent receives \(U_D(e,1)\), where
\begin{equation}
U_D(e,1)
=
B(e,q_D)
-
K_D(e)
-
H(e,N)
+
(N-1)r
+
\chi.
\end{equation}
By the condition in Eq. \eqref{eq:entry_condition}, \(U_D(e,1)>U_0(e)\). Hence, the
allocation with \(k=0\) is not stable.

Now consider the allocation with \(k=1\). The sole decision-maker receives
\(U_D(e,1)\). If this agent exits the central role, the society returns to the
\(k=0\) allocation and the deviating agent receives \(U_0(e)\). By the condition in Eq. 
\eqref{eq:entry_condition}, \(U_D(e,1)>U_0(e)\), so the sole decision-maker
does not want to exit.

It remains to show that no follower wants to become a second decision-maker. A
follower obtains
\begin{equation}
U_F(e)=B(e,q_D)-L(e,N)-K_S(e)-r.
\end{equation}
If the follower deviates and becomes a second decision-maker, the deviating
agent obtains
\begin{equation}
U_D(e,2)
=
B(e,q_D)
-
K_D(e)
-
H(e,N)
-
\eta
+
\frac{N-2}{2}r
+
\chi
-
\psi.
\end{equation}
The condition \(U_F(e) \geq U_D(e,2)\) is equivalent to
\begin{equation}
K_D(e)
+
H(e,N)
+
\eta
+
\psi
-
L(e,N)
-
K_S(e)
-
\chi
\geq
\frac{N}{2}r.
\end{equation}
This inequality is implied strictly by condition
\eqref{eq:no_duplication_condition}. Therefore, no follower wants to enter the
central role when one central decision-maker already exists. The allocation
with \(k=1\) is stable.

Finally, consider any allocation with \(k \geq 2\). We show that at least one
decision-maker wants to exit and become a follower. For any \(k \geq 2\), the
payoff of a decision-maker is
\begin{equation}
U_D(e,k)
=
B(e,q_D)
-
K_D(e)
-
H(e,N)
-
(\eta+\psi)(k-1)
+
\frac{N-k}{k}r
+
\chi.
\end{equation}
The function \(U_D(e,k)\) is strictly decreasing in \(k\) for \(k \geq 1\),
because the duplication and authority-dilution term \((\eta+\psi)(k-1)\)
increases with \(k\), while the per-decision-maker transfer term
\((N-k)r/k\) decreases with \(k\). Therefore, for every \(k \geq 2\),
\begin{equation}
U_D(e,k) \leq U_D(e,2).
\end{equation}
From the condition in Eq. \eqref{eq:no_duplication_condition}, we have
\(U_F(e)>U_D(e,2)\). Hence, for every \(k \geq 2\),
\begin{equation}
U_F(e)>U_D(e,k).
\end{equation}
Thus, any decision-maker in an allocation with \(k \geq 2\) can improve its
payoff by exiting the central role and becoming a follower. Therefore, no
allocation with \(k \geq 2\) is stable.

We have shown that \(k=0\) is unstable, \(k=1\) is stable, and every
\(k \geq 2\) is unstable. Hence, under the stated conditions, the unique stable
equilibrium contains exactly one central decision-maker and \(N-1\) followers.
\end{proof}

Proposition~\ref{prop:unique_central_dm} shows that symmetry at the individual
level does not imply symmetry in equilibrium roles. Even when all agents are
initially identical, the combination of positive returns to being the first
decision-maker and negative returns to duplicating the central role generates
asymmetric specialization. One agent becomes the central decision-maker, while
the remaining agents rationally remain followers. The equilibrium is therefore
asymmetric in roles but symmetric in the sense that any agent could have
occupied the central position ex ante.

\section{Numerical Investigation}
\label{sec:results}
In this section, we first describe the numerical setup and parameterization used to evaluate the theoretical model, and then present the resulting welfare comparisons, dominance regions, transfer-feasibility analysis, role-choice dynamics, and sensitivity tests.

\subsection{Setup and analysis design}
\label{subsec:numerical_setup}
The numerical investigation is designed to illustrate the theoretical mechanisms developed in Section~\ref{sec:model}. The analysis does not aim to estimate the model empirically, but rather to show how the welfare and individual-rationality conditions behave under a transparent parameterization. Specifically, the figures examine five questions. First, how do the three social configurations compare as environmental complexity increases? Second, over which regions of environmental complexity and population size does each configuration dominate? Third, what are the positive and negative components that determine the welfare advantage of decision-compression over universal complexity? Fourth, when does a feasible transfer interval exist such that both followers and the central decision-maker prefer decision-compression? Fifth, does the role-choice game generate a single central decision-maker rather than no decision-maker or multiple competing decision-makers?

Following the theoretical model, the gross benefit of information is specified as
\begin{equation}
B(e,q)=\alpha e q.
\end{equation}
Cognitive costs are given by
\begin{equation}
K_S(e)=c_S e,
\qquad
K_C(e)=c_C e^\beta,
\qquad
K_D(e)=c_D e^\beta,
\end{equation}
and delegation loss and institutional overhead are specified as
\begin{equation}
L(e,N)=\lambda e N^\ell,
\qquad
H(e,N)=h e N^m.
\end{equation}
The baseline parameters are
\begin{equation}
\alpha=3.0,\quad
q_S=0.45,\quad
q_C=0.72,\quad
q_D=0.82,
\end{equation}
\begin{equation}
c_S=0.40,\quad
c_C=0.70,\quad
c_D=0.95,\quad
\beta=1.40,
\end{equation}
and
\begin{equation}
\lambda=0.025,\quad
\ell=0.80,\quad
h=0.020,\quad
m=0.70.
\end{equation}
For the individual-rationality and role-choice analyses, the baseline values are
\begin{equation}
r=0.10,\qquad
\chi=0,\qquad
\eta=0.30,\qquad
\psi=0.30.
\end{equation}
The welfare functions are evaluated over environmental complexity \(e\in[0,10]\) and, where relevant, population size \(N\in[0,250]\). The figures are generated by directly evaluating the analytical expressions for \(W_S(e,N)\), \(W_C(e,N)\), \(W_D(e,N)\), the feasible-transfer bounds \(r_{\min}(e,N)\) and \(r_{\max}(e,N)\), and the role-choice payoffs \(U_D(e,k)\), \(U_F(e)\), and \(U_0(e)\). The sensitivity analysis varies only the delegation-loss parameter \(\lambda\), holding all other parameters fixed, in order to isolate the effect of delegation inefficiency on the dominance region of decision-compression.

\subsection{Results}
\label{subsec:numerical_results}
Figure~\ref{fig:welfare_curves} shows the total welfare of the uniformly simple, uniformly complex, and decision-compression societies as environmental complexity \(e\) increases for a fixed population size. Technically, the figure plots \(W_S(e,N)\), \(W_C(e,N)\), and \(W_D(e,N)\) on the same axis for \(N=110\). The relevant result is that decision-compression can dominate when complexity is high enough for information to be valuable, but when universal complexity becomes costly because every agent independently pays the complex cognitive cost.

\begin{figure}[t]
    \centering
    \includegraphics[width=0.99\textwidth]{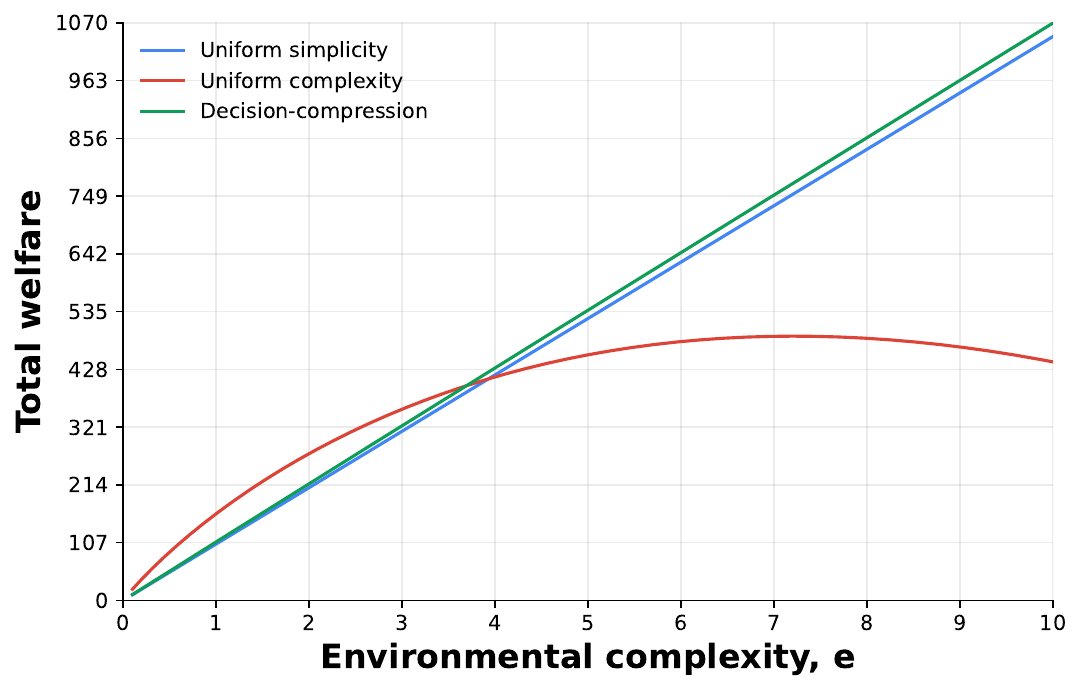}
    \caption{Welfare comparison across environmental complexity. The figure compares uniformly simple, uniformly complex, and decision-compression societies for a fixed population size.}
    \label{fig:welfare_curves}
\end{figure}

Figure~\ref{fig:dominance_phase} shows the welfare-maximizing social configuration over the environmental-complexity and population-size plane. Technically, each point in the \((e,N)\) grid is assigned to the configuration with the highest value among \(W_S(e,N)\), \(W_C(e,N)\), and \(W_D(e,N)\). The relevant result is the emergence of a decision-compression region: as populations become larger and environments become more complex, centralized information processing can outperform both universal simplicity and universal complexity.

\begin{figure}[t]
    \centering
    \includegraphics[width=0.99\textwidth]{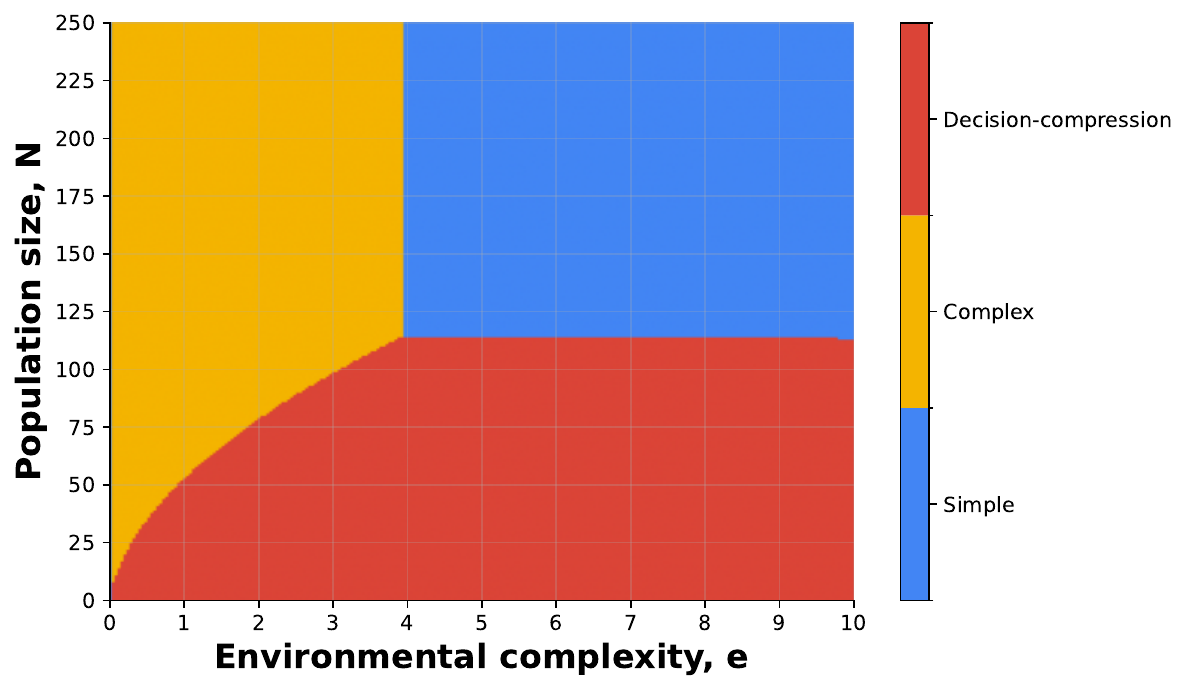}
    \caption{Dominant social configuration in the \((e,N)\) plane. Each point shows which configuration produces the highest total welfare.}
    \label{fig:dominance_phase}
\end{figure}

Figure~\ref{fig:decomposition} decomposes the welfare advantage of decision-compression over universal complexity. Technically, the figure plots the informational advantage, the cognitive-cost savings, the delegation and institutional costs, and the total difference \(W_D(e,N)-W_C(e,N)\). The relevant result is that decision-compression does not dominate simply because it uses better information; it dominates when the savings from avoiding duplicated complex cognition are large enough to compensate for delegation loss, institutional overhead, and the central decision-maker's cost.

\begin{figure}[t]
    \centering
    \includegraphics[width=0.99\textwidth]{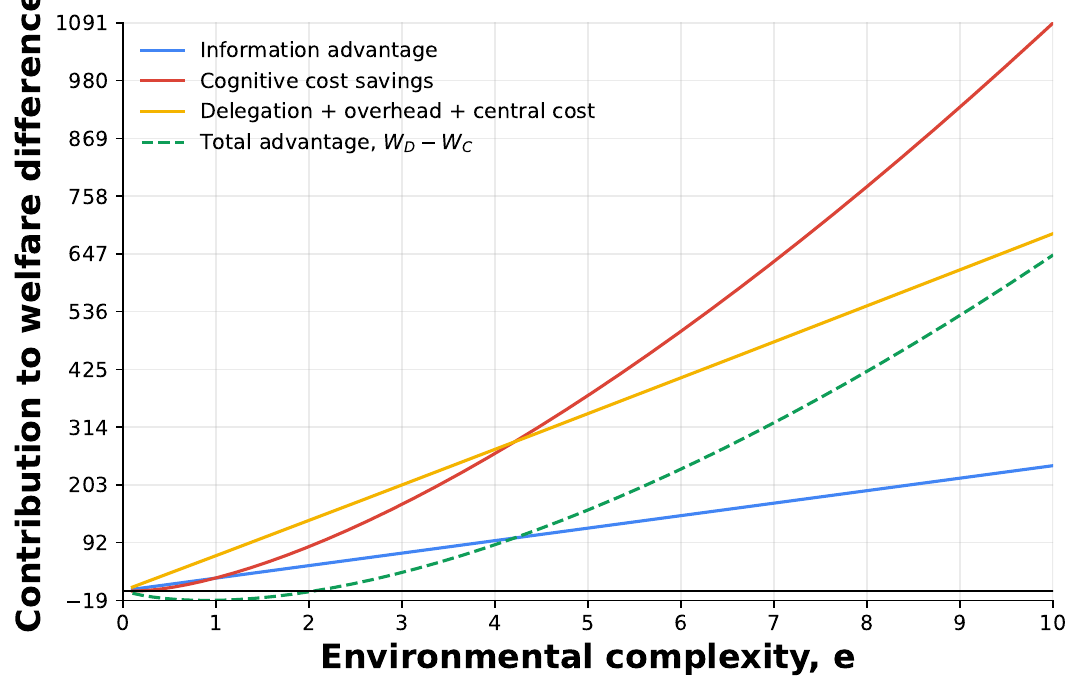}
    \caption{Decomposition of the welfare advantage of decision-compression over universal complexity. The dashed line shows the total welfare difference \(W_D-W_C\).}
    \label{fig:decomposition}
\end{figure}

Figure~\ref{fig:transfer_interval} shows the feasible transfer interval that can sustain decision-compression as an individually rational arrangement. Technically, the lower curve is the minimum transfer \(r_{\min}(e,N)\) required for the central decision-maker to accept the role, while the upper curve is the maximum transfer \(r_{\max}(e,N)\) followers are willing to pay. The relevant result is that there are parameter regions in which \(r_{\min}(e,N)\leq r_{\max}(e,N)\), so the central decision-maker is not a volunteer who sacrifices utility for the group; rather, the institution can be mutually sustainable.

\begin{figure}[t]
    \centering
    \includegraphics[width=0.99\textwidth]{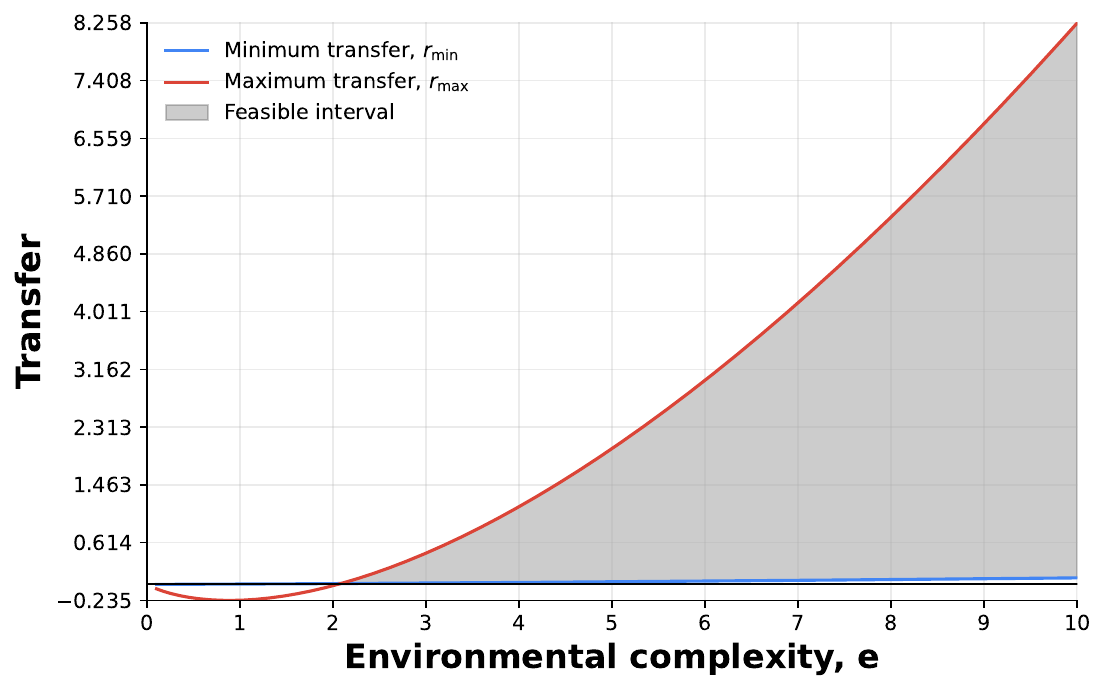}
    \caption{Feasible transfer interval for decision-compression. The shaded region indicates values of \(r\) that satisfy both the central decision-maker's and the followers' participation constraints.}
    \label{fig:transfer_interval}
\end{figure}

Figure~\ref{fig:role_choice} shows the payoff structure of the role-choice game as the number of decision-makers \(k\) changes. Technically, the figure plots the decision-maker payoff \(U_D(e,k)\) together with the follower payoff \(U_F(e)\) and the complex baseline payoff \(U_0(e)\). The relevant result is that becoming the first decision-maker can be attractive, but additional decision-makers reduce transfer revenue and create duplication or authority-dilution costs. This supports the theoretical result that exactly one central decision-maker can emerge as the stable role allocation.

\begin{figure}[t]
    \centering
    \includegraphics[width=0.99\textwidth]{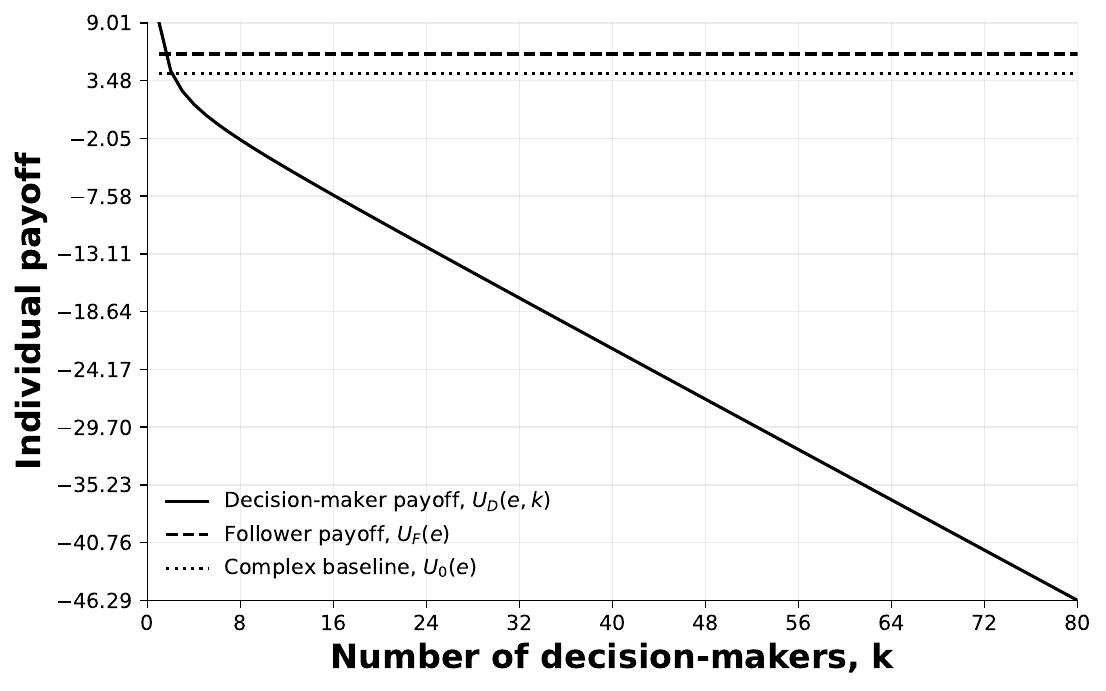}
    \caption{Role-choice payoffs as the number of decision-makers increases. The different line styles show the payoff of becoming a decision-maker, remaining a follower, and receiving the complex baseline payoff.}
    \label{fig:role_choice}
\end{figure}

Figure~\ref{fig:sensitivity_lambda} shows the sensitivity of the dominance regions to the severity of delegation loss. Technically, the four panels repeat the \((e,N)\) dominance diagram for \(\lambda\in\{0.005,0.015,0.025,0.045\}\), holding all other parameters fixed. The relevant result is that decision-compression is most robust when delegation loss is low or moderate. As \(\lambda\) increases, the decision-compression region contracts, showing that the advantage of centralized cognitive processing depends on the institution's ability to translate complex information into usable guidance without excessive distortion.

\begin{figure}[t]
    \centering
    \begin{subfigure}[t]{0.48\textwidth}
        \centering
        \includegraphics[width=\textwidth]{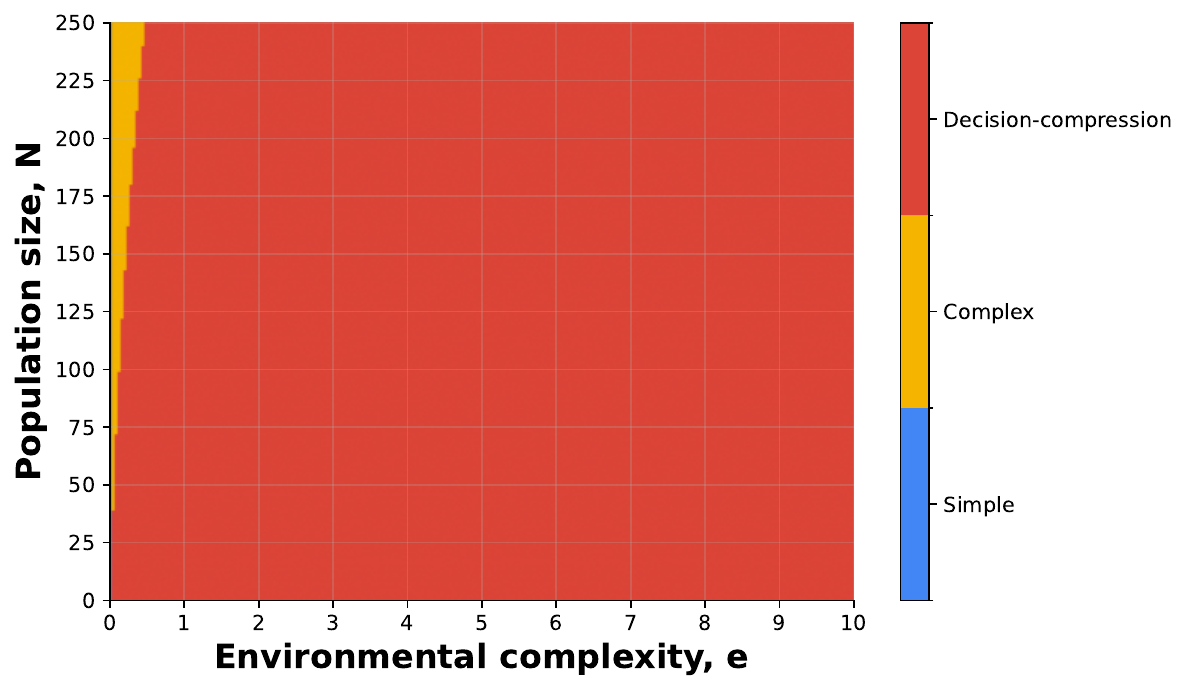}
        \caption{\(\lambda=0.005\)}
        \label{fig:sensitivity_lambda_0005}
    \end{subfigure}
    \hfill
    \begin{subfigure}[t]{0.48\textwidth}
        \centering
        \includegraphics[width=\textwidth]{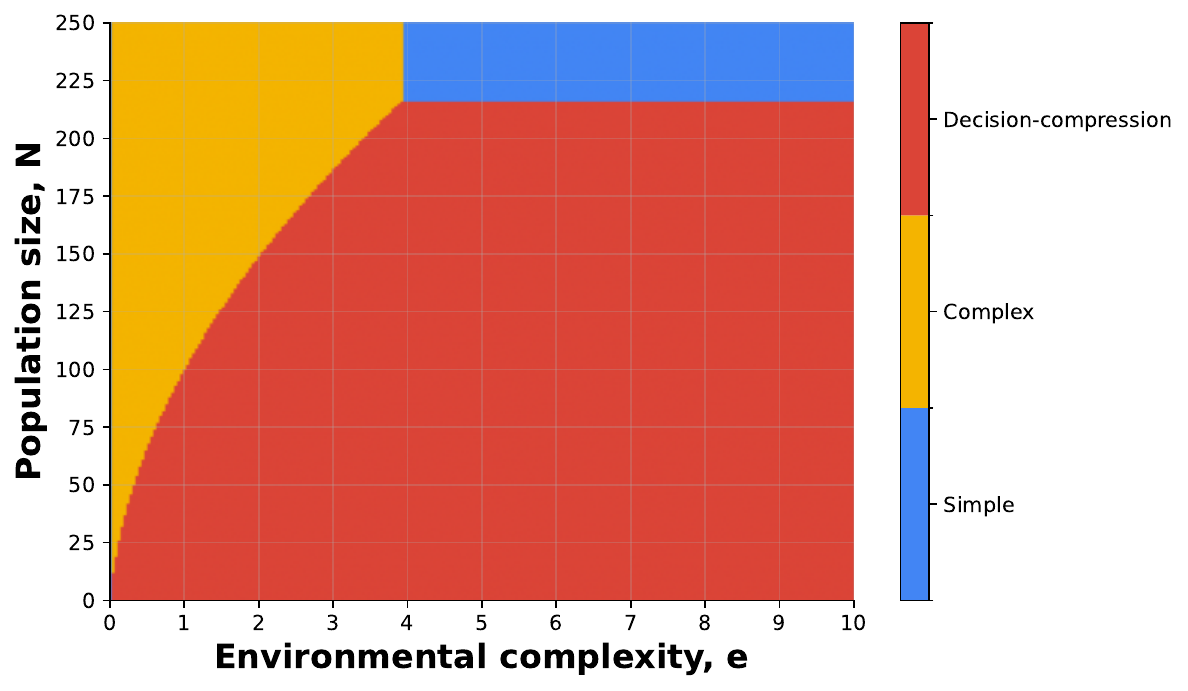}
        \caption{\(\lambda=0.015\)}
        \label{fig:sensitivity_lambda_0015}
    \end{subfigure}

    \vspace{0.4cm}

    \begin{subfigure}[t]{0.48\textwidth}
        \centering
        \includegraphics[width=\textwidth]{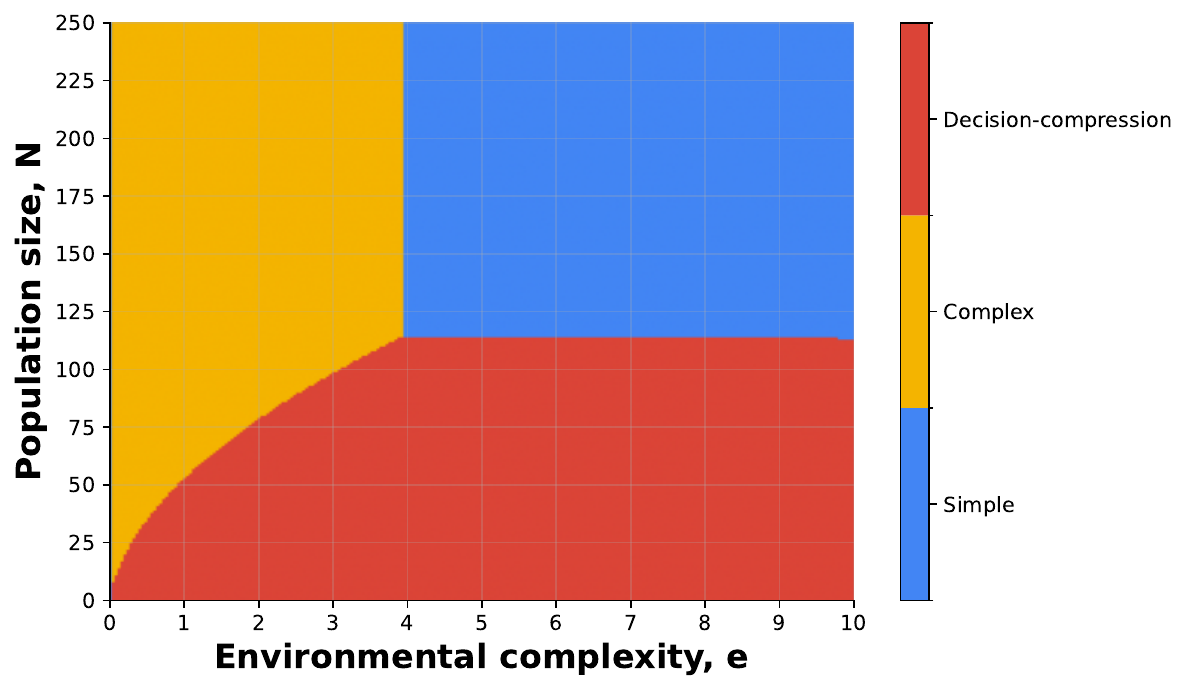}
        \caption{\(\lambda=0.025\)}
        \label{fig:sensitivity_lambda_0025}
    \end{subfigure}
    \hfill
    \begin{subfigure}[t]{0.48\textwidth}
        \centering
        \includegraphics[width=\textwidth]{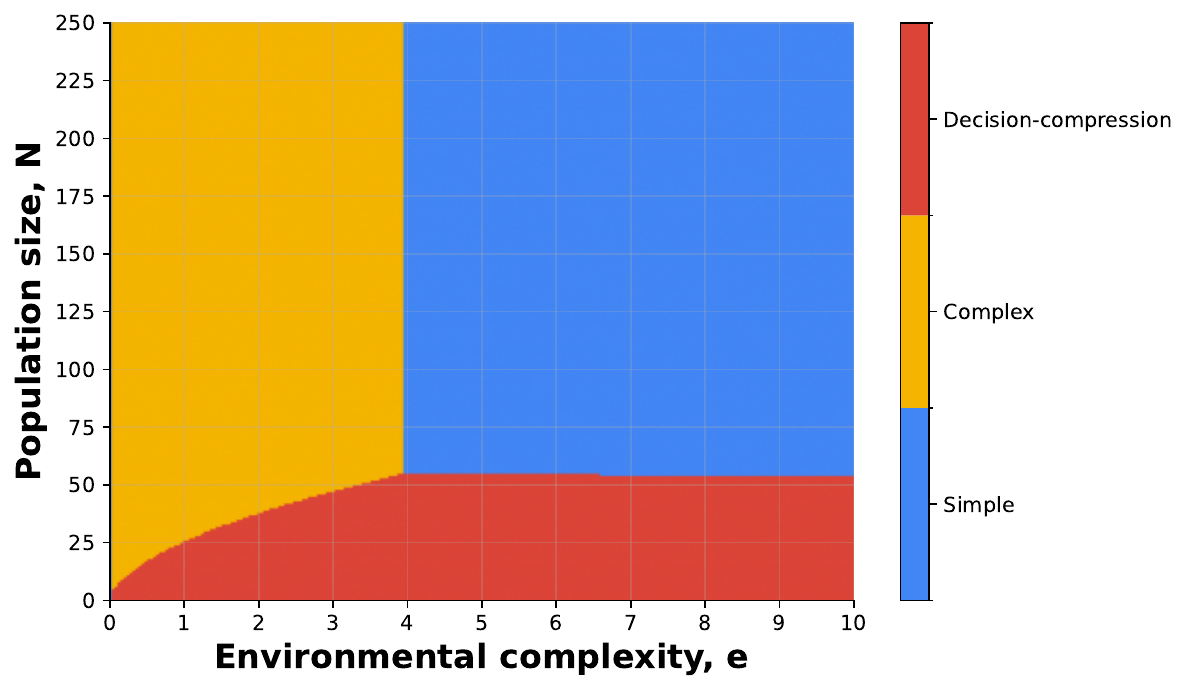}
        \caption{\(\lambda=0.045\)}
        \label{fig:sensitivity_lambda_0045}
    \end{subfigure}
    \caption{Sensitivity of dominance regions to delegation-loss severity. Each panel shows the welfare-maximizing social configuration in the \((e,N)\) plane for a different value of \(\lambda\).}
    \label{fig:sensitivity_lambda}
\end{figure}

\section{Discussion}
\label{sec:discussion}
In this study, we address the question \say{if information improves decisions, why do populations not evolve toward universal cognitive complexity?} by treating cognition as an economically costly input rather than as a free improvement in decision quality. The main outcome is that decision-compression can dominate both universal simplicity and universal complexity. Universal simplicity saves cognitive cost but loses informational value, whereas universal complexity preserves autonomy and high-quality decision-making but requires every agent to pay the full cost of information processing. A heterogeneous structure, in which most agents remain cognitively simple while a specialized decision-making centre processes complex information, can outperform both extremes when the cognitive costs saved by the population exceed delegation loss, institutional overhead, and the central decision-maker's cost.

The meaning of the model is therefore not that centralization is always efficient, nor that simplicity is always adaptive. Rather, the model identifies a general allocation problem: where should costly cognition be located in a society? This connects the paper to economic theories of organization in which firms, hierarchies, and institutions exist partly because market exchange and individual optimization are not costless \cite{Coase1937}. Transaction-cost economics similarly argues that organizational form depends on the comparative cost of coordinating activity through markets, contracts, or authority \cite{Williamson1975}. From this perspective, decision-compression institutions are not deviations from rationality but devices for economizing on the repeated cost of information acquisition and interpretation. They transform complex environments into simpler operational signals, much as organizations transform uncertainty into routines, commands, budgets, and procedures \cite{Arrow1974}.

The model also clarifies why social complexity need not require universal individual complexity. Large-scale systems often work because information is filtered, summarized, and routed through specialized roles. Economic systems differ not only in how much information they contain, but also in their architecture: who observes what, who processes what, and who communicates decisions to whom \cite{SahStiglitz1986}. Hierarchies can therefore be interpreted as information-processing structures that reduce the burden placed on each individual agent \cite{Radner1993}. In the present model, this logic appears formally in the comparison between duplicated cognition and centralized cognition. When many agents independently solve similar complex problems, society pays the same cognitive cost many times. Decision-compression reduces this duplication by concentrating costly processing in one role, while distributing simplified outputs to the rest of the population.

The numerical results support the analytical results but also refine them. Analytically, the dominance condition shows that decision-compression is beneficial when informational gains and cognitive savings exceed delegation and institutional costs. Numerically, the welfare curves show how this condition can emerge gradually as environmental complexity increases. At low levels of complexity, universal simplicity can perform well because little is lost by relying on low-cost cognition. At higher levels of complexity, however, the value of information increases, and simple decision-making becomes increasingly inadequate. Universal complexity may then become attractive, but only until the cost of duplicated cognition becomes too large. The numerical results therefore refine the theory by showing that decision-compression is most naturally interpreted as a middle regime: it is most valuable when information is sufficiently important to justify specialized processing, but cognitive costs are sufficiently high to make universal complexity inefficient. Moreover, the phase diagrams further refine the analytical story by showing that population size changes the relative attractiveness of the three social structures. In the analytical model, population size enters through both cognitive savings and delegation losses. The numerical investigation makes this trade-off visible. As \(N\) increases, the potential saving from avoiding duplicated complex cognition grows because more agents can rely on the same compressed output. At the same time, larger populations can increase communication burdens, monitoring costs, and distortion. This means that decision-compression is not simply a large-population result; it is a large-population result only when delegation and institutional overhead scale moderately. This point is consistent with organizational theories showing that decentralization and centralization each have costs, and that efficient design depends on the informational structure of the task \cite{MarschakRadner1972}. In addition, the role-choice analysis contributes a further implication: asymmetry in social roles can emerge even when agents are ex ante symmetric. No individual is assumed to be naturally destined to become the decision-maker. Instead, the first central role can be attractive because it captures concentrated benefits, while additional decision-makers reduce rents and create duplication or conflict. This resembles economic models in which organizational positions emerge endogenously from returns to specialization and coordination rather than from fixed biological or social differences \cite{Becker1992}. It also relates to theories of authority in which command structures persist because they reduce coordination costs, but become inefficient when authority is duplicated or contested \cite{MilgromRoberts1992}. Thus, the model provides a mechanism by which unequal cognitive roles can arise from initially symmetric agents.

This study is not without limitations. First, the proposed model represents cognition through two discrete modes, simple and complex, whereas real agents may differ continuously in ability, attention, experience, and access to information. Second, the central decision-maker is modeled as a single role, but real societies often contain multiple overlapping centres of expertise, including firms, professions, bureaucracies, markets, courts, media systems, and algorithmic platforms. Third, the model abstracts from dynamic learning. In reality, followers may become more capable over time, central institutions may improve or degrade, and delegation losses may depend on trust, legitimacy, and feedback. Fourth, the numerical investigation is illustrative rather than empirical. It shows that the theoretical conditions can generate meaningful regimes, but it does not estimate the size of cognitive costs, delegation losses, or institutional overhead in a specific population. Finally, the model does not fully address political economy problems such as rent extraction, capture, manipulation, and the strategic production of ignorance, all of which may cause decision-compression institutions to persist even when they are no longer socially efficient \cite{Stigler1971}.

Taken jointly, the analytical and numerical results support a broader interpretation of simplicity as an organized economic response to complexity. Populations do not need every agent to process all available information independently. They need mechanisms that determine when cognition should be individualized, when it should be centralized, and how complex outputs should be translated into forms that ordinary agents can use. Decision-compression is beneficial only under specific conditions, but those conditions are theoretically important: they show why simplified cognition, specialized expertise, hierarchy, and institutions can coexist with the economic value of information. The simplicity paradox is therefore not a contradiction. It is a consequence of cognitive economy. Namely, when information is valuable but costly, scalable societies may evolve not by making everyone complex, but by organizing complexity.

\section*{Declarations}
\subsection*{Funding}
This study received no funding. 

\subsection*{Conflicts of interest/Competing interests}
None.

\subsection*{Acknowledgement}
The author wishes to thank Labib Shami for helping shape the idea of this study and its presentation.

\subsection*{Usage of AI tools}
The author used AI tools for the manuscript preparation and for the numerical code development. All materials were manually reviewed, and the author takes full responsibility for the final content.

\bibliography{biblio}
\bibliographystyle{unsrt}

\end{document}